\documentclass[12pt]{article}
\usepackage{amsmath, amsthm, amscd, amsfonts, amssymb, graphicx, color}
\usepackage[bookmarksnumbered, colorlinks, plainpages]{hyperref}
\hypersetup{colorlinks=true, linkcolor=red, anchorcolor=green, citecolor=green, urlcolor=black, filecolor=black, pdftoolbar=true}
\newtheorem{theorem}{Theorem}[section]
\newtheorem{lemma}[theorem]{Lemma}

\theoremstyle{definition}
\newtheorem{definition}[theorem]{Definition}

\theoremstyle{remark}
\newtheorem{remark}{Remark}[section]
\numberwithin{equation}{section}

\begin{document}

	\setcounter{page}{1}
	
	
	\begin{center}
		{\large  \bf{On some properties of M-projective curvature tensor in spacetime of general relativity}}
		
		\bigskip
		
	Musavvir Ali,$^{1,a}$ 	Mohammad Salman,$^{1,b}$ Farook Rahaman$^{2,*}$  and Naeem Ahmad Pundeer,$^{2,c}$   \\

	$^{1}$ Department of Mathematics, 
	
	Aligarh Muslim University, Aligarh-202002, India.

	$^{2}$ Department of Mathematics, 

Jadavpur University, Kolkata-700032, West Bengal, India.

	\footnote{Email addresses:  $^{a}$\url{musavvir.alig@gmail.com}, $^{b}$\url{salman199114@gmail.com}, $^{*}$\url{farookrahaman@gmail.com}, $^{c}$\url{pundir.naeem@gmail.com}   }

	\end{center}
	
	\bigskip
	
	\begin{abstract}
\noindent In this paper, we investigate the connection between the M-projective curvature tensor and other tensors. Also, we obtain the divergence of M-projective curvature tensor. A symmetry of spacetime known as M-projective collineation has been presented, and it has been possible to determine the conditions under which the general relativity spacetimes can admit such collineations.

	\end{abstract}
	
	\noindent{\small {\bf  MSC2020:} 53C50; 53C25; 83C20; 83C50; 83C05. }

	\section{ Introduction}
	
A manifold is the foundation of the most recent theories of general relativity in mathematics. The Lorentzian manifold  \cite{Besse,Neill} structure is the most significant subclass of pseudo-Riemannian manifolds after Riemannian manifolds, and these have important applications in cosmology and general relativity \cite{K.L.DuggalandR.Sharma, hall, KarmerD, yanoK}. The curvature tensor is a crucial element of differential geometry. 	The study of its symmetries has gained a lot of attention recently. The potential infinite dimensionality of the related vector space of curvature inheritance makes it a difficult task. This study will make a contribution to the researchers by identifying further instances of these symmetries across a wide variety of spacetimes,  which can make the task easier.
\vspace{6pt}
\par

The curvature of spacetime and the dispersion of matter inside it are described by general relativity, which uses the field equations structure. What we experience as gravity is the interaction between matter and spacetime? Although the spacetime continuum idea already existed, general relativity allowed Einstein to explain gravity as a bending of spacetime geometry. Einstein developed a set of field equations to describe how gravity reacted to matter in spacetime.	The primary goal of the research is to design potentials that correspond to Einstein's field equations. Symmetries of spacetime  on the geometry are consistent with the dynamics of the selected distribution. Geometrical  symmetries of spacetime can be expressed by the vanishing of the Lie derivative of certain tensors along a vector filed, these vector field may be  null, spacelike, or timelike. A number of works, Katzin, Levine, and Davis proposed the concept of symmetries
in the general theory of relativity  
 \cite{Katzin4, Katzin, Katzin2}. In \cite{Zahsan2,Ahsan,ZAhsan3,Ahsan2,ZAhsan2,AhsanSA,Zahsan12}, Ahsan et al. carried out additional research on the general relativity's spacetime symmetries. Y. J. Suh and his collaborator \cite{PundeerSuh,SKshu2,Suh2,Suh,YJSuh2} have also studied as  further research in context of spacetime of general relativity and Kaehler manifold.
\vspace{6pt}

The notation is conventional, with $V_4$ denoting the (four-dimensional, linked, Hausdorff) spacetime manifold and $g$ denoting the Lorentzian  metric with the signature $(-, +,+,+)$. The Ricci tensor components are $R _{ij} = R^ h_ {ihj}$, the Ricci scalar is $R = g^{ij} R_ {ij}$, and the curvature tensor associated with $g$ is indicated in component form by $R^h_ {ijk}$.  The symbols $(,)$,  $\pounds_ \xi$,  and $(;)$  represents the partial derivatives, covariant derivatives and Lie derivatives, respectively. The traditional symmetrization and skew-symmetrization are indicated by round and square brackets, respectively. Finally, it is assumed that $V_4$ is not flat in the sense that the curvature tensor does not vanish across a non-empty open subset of $V_4$.
\vspace{6pt}
\par

	The interaction between the matter and free gravitational parts of the field is described by the Bianchi identities.  The construction of the gravitational potential is the main objective of all investigations which satisfies the Einstein field equations (EFEs, briefly) with a cosmological constant. There is an interaction of matter and gravitation through EFE which is given by  \cite{KarmerD}
	\begin{equation}\label{2.1}
		R_{ij} - \frac{R}{2} g_{ij} + \wedge g_{ij} =  k ~  T_{ij},
	\end{equation}
	
	\noindent  where  $ g_{ij}$,  $R_{ij}$, $T _{ij}$, $\wedge$ and  $k$  denote metric,  Ricci tensor, energy-momentum tensor, the cosmological term and   constant, respectively. The energy momentum tensor is described in the equation,
	\begin{equation}\label{2.2}
		T_{ij}=(\mu + p) u_{i}u_{j}+ p g_{ij},
	\end{equation} 
	here $p$  the isotropic pressure,    $\mu$  the energy density, and $u^{i}$ is the velocity vector field   of the flow satisfying  $g_{ij} u^{j} = u_{i}$  for all  $i$  and $u_{i}u^{i}= -1$. 
	\vspace{6pt}
	
	\par	
	
	\section{Preliminaries}
	
	The objective is accomplished when we put symmetry assumptions on the geometry that are compatible with the dynamics of the selected matter distribution. The following equation represents these geometrical symmetries \cite{yanoK} of spacetime,
	\begin{equation}\label{2.3}
		\pounds_\xi  A =2  \varOmega A
	\end{equation}
	where $ \pounds_\xi $ stands for the Lie derivative along the vector field $\xi^{i}$. The vector field $\xi^{i}$ may be time-like ($\xi^{i}\xi_{i}<0$), space-like ($\xi^{i}\xi_{i}>0$) or null ($\xi^{i}\xi_{i}=0$),  \textquoteleft A\textquoteright  ~~ denotes a geometrical/physical quantity and $\varOmega$  is a scalar function.
	\vspace{6pt}
	\par
	
	A simple example can be provided here as the metric inheritance symmetry or \cite{yanoK} the  conformal motion (Conf M) along a conformal Killing vector (CKV) $\xi^{i} $ for A = $g_{ij}$  in equation  \eqref{2.3}  
	\begin{equation}\label{2.4}
		\pounds_\xi g_{ij}   =  2 \varOmega g_{ij} 
	\end{equation} 
	
	\noindent where $\varOmega$ is a conformal function and  the   most primary symmetry on  $(V_4, g) $ is the  motion (M),  which is obtained by setting $ A = g_{ij} $ and $\varOmega = 0$ in equation \eqref{2.3} and $\xi^{i}$  is called  a Killing vector \cite{yanoK,AA Coley,Michalaski }.
	\vspace{6pt}
	\par

The literature on these collineations and geometrical symmetry is growing, with conclusions that are elegant. More than 30 different forms of collineations have been examined.  We will only discuss symmetry assumptions that are necessary for our research.
\vspace{6pt}
\par	
	
The concept of curvature inheritance was introduced by K.L. Duggal in 1992  \cite{Duggal}.  Curvature inheritance is the generalization of curvature collineation (CC),   which was  defined in 1969 by Katzin  \cite{Katzin}. Most recently, Ali et al. \cite {AliM,AliM3,AliM2,Sheikh} investigated symmetry inheritance for conharmonic curvature tensors and Ricci flat spacetimes.
	\begin{definition}
		\cite{Duggal} A spacetime $(V_4, ~g)$ admits curvature inheritance symmetry along a 	 smooth vector field $\xi^{i}$ such that 
		
		\begin{equation} \label{2.5}
			\pounds_\xi R^h_{ijk} = 2 \varOmega R^h_{ijk}, 
		\end{equation}
		where $\varOmega = \varOmega(x^{i})$ can be called as inheriting factor or inheritance function.
	\end{definition}
	
	\noindent If  $\varOmega = 0$, then $ 	\pounds_\xi R^h_{ijk} = 0$ and $\xi^{i}$ is said to follow a curvature symmetry \cite{Katzin4,Katzin,Katzin2}  on $V_4$ or simply to write it generate a curvature collineation (CC) .
	The CI equation \eqref{2.5}	 can be written in local expression
	\begin{equation}\label{2.6}
		R ^h_{ijk; l} \xi^l - R^l _{ijk} \xi^h_{;l} + R^h_{ljk} \xi ^l_{;i} +R^h_{ilk} \xi ^l_{;j} +R^h_{ijl} \xi ^l_{;k} = 2 \varOmega  R^h_{ijk}.
	\end{equation}
	
	\begin{definition}
		\cite{Duggal}	A spacetime $(V_4, ~g)$ admits Ricci inheritance (RI) symmetry along a 	 smooth vector field $\xi^{i}$ such that  
		
		\begin{equation} \label{2.7}
			\pounds_\xi R_{ij} = 2 \varOmega R_{ij}.
		\end{equation}
		
	\end{definition}
	\noindent Contraction of \eqref{2.5} implies to \eqref{2.7}. Thus, in general, every curvature inheritance is  Ricci inheritance symmetry (i.e., CI $\Rightarrow$ RI), but the converse may not hold. In particular, RI reduces to RC when $\varOmega =0$. Otherwise, for $\varOmega \neq 0$, it is called a proper RI.

\begin{definition}
A spacetime admits  motion/isometry along a vector field $\xi^{i}$  such that
	\begin{equation}\label{1.8}
	\pounds_\xi g_{ij}   =  0 = \xi_{i;j} + \xi_{j:i}. 
\end{equation} 	
Equation \eqref{1.8} is referred to as the Killing equation, and vector $\xi$ is referred to as the Killing vector field.
\end{definition}

\begin{definition}
	A spacetime  admits   conformal motion (Conf M) along a vector field $\xi^{i}$  such that
	\begin{equation}\label{1.9}
		\pounds_\xi g_{ij}   =  2 \psi g_{ij}.  
	\end{equation} 	
Vector $\xi$ is referred to as a conformal Killing vector field, and equation \eqref{1.9} is known as a conformal Killing equation.
\end{definition}
\noindent	The conformal curvature tensor $\cite{Weyl}$ expressed as follows 
	\begin{equation}\label{2.8}
		C^h_{ijk}= R^h_{ijk}+\frac{1}{2}(\delta^h_j R_{ik}-\delta^h_k R_{ij}+ R^h_j g_{ik}- R^h_k g_{ij})+\frac{R}{6}
		(g_{ij}\delta^h_k- g_{ik}\delta^h_j).
	\end{equation}
	\begin{definition}
	 A spacetime $(V_4, ~g)$ admits Weyl curvature collineation (WCC) along a smooth vector field $\xi^{i}$ such that
	
	\begin{equation} \label{1.13}
		\pounds_\xi C^h_{ijk} =0.
	\end{equation}
\end{definition}
\begin{definition}
	 A spacetime $(V_4, ~g)$ admits projective  curvature collineation  along a smooth vector field $\xi^{i}$ such that
	
	\begin{equation} \label{1.14}
		\pounds_\xi W^h_{ijk} =0,
	\end{equation}
\end{definition}
\noindent where $W^h_{ijk} = R^h_{ijk} - \frac{1}{3} (\delta ^h_{k} R_{ij}  - \delta ^h_{j} R_{ik})$ is the projective curvature tensor.

\begin{definition}
The symmetry property of spacetime determined by the electromagnetic field is
	\begin{equation} \label{1.15}
		\pounds_\xi F_{ij} =  F_{ij;k} \xi^{k} + F_{ik} \xi^k_{;j} + F_{jk} \xi^{k}_{;i} = 0
	\end{equation}
\end{definition}
\noindent where $F_{ij} $ is the electromagnetic field tensor.
\vspace{6pt}
\par

Motivated by the above works, we have  establish relationships among the divergences of the M-projective, projective, conformal, conharmonic, concircular and W-curvature tensors. Also,  we  introduce the M-projective collineation. We have proved the results for divergences of M-projective curvature tensor in Section 3. In Section 4, we will discuss the  study of M-projective curvature collineation and related results. An  overview of the study is provided in Section 5 to conclude.

	\section{\small Divergence of M-projective curvature tensor and its relationship with divergence of  other curvature tensors }

 In this section, we determine the relations between the divergence of the M-projective curvature tensor and some other curvature tensors. Also,  we will describe the M-projective curvature tensor in terms of conformal,   conharmonic, concircular, projective   and W-curvature tensors.
 \vspace{6pt}
 \par
	
The properties of an M-projective curvature tensor were proposed by Pokhariyal and Mishra  \cite{GPPokh} in 1970. This tensor is described as follows
	\begin{eqnarray}\label{2.1n}
\bar{W} (X,Y,Z,T) = \bar{R} (X,Y,Z,T) - \frac{1}{2(n-1)} [S(Y,Z) g(X,T) - S(X,Z) g(Y,T)\notag\\ + g(Y,Z) S(X,T)  - g(X,Z) S(Y,T)] ,
\end{eqnarray}
where   $ \bar{W} (X,Y,Z,T) = g(W(X,Y)Z, T)$ and  $ \bar{R} (X,Y,Z,T) = g(R(X,Y)Z, T)$.  The local coordinates expression of equation \eqref{2.1n} are as follows
\begin{eqnarray}\label{2.2n}
	\bar{W}_{ijkl} = R_{ijkl}- \frac{1}{2(n-1)} [R_{jk} g_{il} - R_{ik} g_{jl}  + g_{jk} R_{il}  - g_{ik} R_{jl}].
\end{eqnarray}
Where $S(X,Y)$ is the global form of the Ricci tensor.The following characteristics hold for this tensor,
\begin{equation}\label{2.3n}
	\bar{W}_{ijkl} = - \bar{W}_{jikl}, ~~  \bar{W}_{ijkl} = - \bar{W}_{ijlk}, ~~~ \bar{W}_{ijkl} =  \bar{W}_{klij}, 
\end{equation}

\begin{equation}\label{2.4n}
	\bar{W}_{ijkl} + \bar{W}_{jkil}  + \bar{W}_{kijl} = 0.
\end{equation}
Assuming $n = 4$ in equation \eqref{2.2n} and contracting with $g^{ih}$, the M-projective curvature tensor for the spacetime of general relativity is given by
\begin{eqnarray}\label{2.5n}
	\bar{W}^h_{jkl} = R^h_{jkl} - \frac{1}{6} [R_{jk} \delta^h_{l} - R_{jl} \delta^h_{k} + g_{jk} R^h_{l}  - g_{jl} R^h_{k}].
\end{eqnarray}

	\begin{definition}
	A Riemannian space is  M-projectively flat  $\cite{Weyl}$ iff 
	
	\begin{equation}\label{2.9}
		\bar{W}^h_{ijk}=0.  ~~~~  
	\end{equation}
\end{definition} 
\begin{definition}
	An   Einstein spacetime defined as follows 
	\begin{equation}
		R_{ij} = \mu g_{ij},
	\end{equation}
	where $\mu = \frac{R}{4}$
	is a scalar and R is the scalar curvature. If the Ricci tensor is simply proportional to the metric tensor, we say that  this spacetime as Einstein-like spacetime.  
\end{definition} 
\par U. C. De and S. Mallick have stated in \cite{De} \textquotedblleft An M-projectively flat space of dimension four is an  Einstein spacetime and of constant curvature." The spaces of constant curvature play a significant role in cosmology.  The simplest cosmological model of the universe is obtained by making the assumption that the universe is isotropic and homogeneous. This is
known as cosmological principle. By isotropy we mean that all spatial directions are equivalent, while homogeneity means that it is impossible to distinguish one place in the universe from the other. That is, in the rest system of matter there is no preferred
point and no preferred direction, the three dimensional space being constituted in the same way everywhere. This cosmological principle, when translated into the language of Riemannian geometry, asserts that the three dimensional position space is a space of maximal symmetry. A space is of maximal symmetry if it has maximum 
number of Killing vector fields; and the maximum number of Killing vector fields in a Riemannian space of dimension n is $\frac{n (n+1)}{2}$, that is, a space of constant curvature whose curvature depends upon time.
\vskip 6pt

\par
Many authors studied the properties and applications of M-projective curvature tensor, especially by Ojha \cite{RH1,RH2} for Sasakian manifold, while for generalized Sasakian space form Venkatesha et al. \cite{Venkatesha} have studied this tensor. The role of M-projective curvature tensor in the study of $(LCS)_{2n+1}$ manifolds has been investigated by Shanmukha and Sumangala \cite{Shanmukha} while LP-Sasakian Manifolds satisfying the condition $\bar{W}^h_{ijk}$ = 0 have been considered by Zengin \cite{Zengin2}. Most recently, Kishor et al. \cite{Shyam} have studied specific results on M-projective curvature tensor $(\kappa, \mu)$- contact space forms. Chaubey and Ojha \cite{SKRH} have studied the M-projective curvature tensor for the Kenmotsu manifold. Also, lots of authors have studied the M-projective curvature tensor in the context of general relativity \cite{De, Debnath, Debnat2, Zengin1}.
\begin{theorem}
For spacetime $(V_4, g)$ an M-projective curvature tensor is divergence-free if and only if the Ricci tensor is of Codazzi type.
\end{theorem}
\begin{proof}
The M-projective curvature tensor for the spacetime of general relativity is given by
\begin{eqnarray}\label{2.5n1}
	\bar{W}^h_{jkl} = R^h_{jkl} - \frac{1}{6} [R_{jk} \delta^h_{l} - R_{jl} \delta^h_{k} + g_{jk} R^h_{l}  - g_{jl} R^h_{k}].
\end{eqnarray}	
Taking covariant derivative of \eqref{2.5n1}, we get
\begin{eqnarray}\label{2.5n2}
	\bar{W}^h_{jkl;h} = R^h_{jkl;h} - \frac{1}{6} [R_{jk;h} \delta^h_{l} - R_{jl;h} \delta^h_{k} + g_{jk} R^h_{l;h}  - g_{jl} R^h_{k;h}],
\end{eqnarray}
where semicolon $``(;)"$ denotes covariant differentiation.

\noindent We know that
\begin{equation}\label{2.7n}
	R^h_{jkl;h} =     R_{jk;l} - R_{jl;k}
\end{equation}
and 
\begin{equation}\label{2.7n1}
	R^h_{l;h} =    \frac{1}{2} R_{,l},
\end{equation}
where comma  “(,)” denotes partial differentiation. Using \eqref{2.7n} and \eqref{2.7n1} in \eqref{2.5n2}, we have
\begin{eqnarray}\label{2.5n3}
	\bar{W}^h_{jkl;h} = \frac{5}{6}  [R_{jk;l}  - R_{jl;k}] -  \frac{1}{12} [g_{jk} R_{,l}  - g_{jl} R_{,k}].
\end{eqnarray}
If $\bar{W}^h_{jkl;h}$ is divergence free, then the equation \eqref{2.5n3} becomes
\begin{eqnarray}\label{2.5n4}
	 5[R_{jk;l}  - R_{jl;k}] = \frac{1}{2} [g_{jk} R_{,l}  - g_{jl} R_{,k}].
\end{eqnarray}
Multiplying \eqref{2.5n4} by $g^{jk}$, we obtain
\begin{eqnarray}\label{2.5n5}
	5[R_{,l}  - R^k_{l;k}] = \frac{3}{2} R_{,l}.
\end{eqnarray}
Since $R^k_{l;k} = \frac{1}{2} R_{,l},$ the above equation implies $R$ = constant.

\noindent Hence from \eqref{2.5n4}, we get
\begin{eqnarray}\label{2.5n6}
R_{jk;l}  - R_{jl;k} = 0.
\end{eqnarray}
This implies that the Ricci tensor is of Codazzi type \cite{Derdz}.

\noindent Conversely, if the Ricci tensor is of Codazzi type, then
\begin{eqnarray}\label{2.5n7}
	R_{jk;l}  - R_{jl;k} = 0.
\end{eqnarray}
Multiplying \eqref{2.5n7} by $g^{jk}$,  we have 
\begin{equation}\label{2.5n8}
	R_{,} = 0.
\end{equation}
From \eqref{2.5n3}, \eqref{2.5n7} and \eqref{2.5n8}, we obtain
\begin{equation}\label{2.5n9}
	\bar{W}^h_{jkl;h} = 0.
\end{equation}
This implies that the M-projective curvature tensor is divergence-free.

\noindent Therefore the M-projective curvature tensor is divergence-free if and only if the Ricci tensor is of Codazzi type.

\vspace{6pt}

\noindent This completes the proof.
\end{proof}
\begin{theorem}
	For  a spacetime $(V_4, g)$ with constant curvature, the divergences of M-projective curvature tensor and projective curvature tensor are identical.
\end{theorem}

\begin{proof}

The projective curvature tensor $P^h_{ijk}$ is defined for a Riemannian space $V_ 4$ as
\begin{eqnarray}\label{2.10n}
	{P}^h_{ijk} = R^h_{ijk} -  \frac{1}{3} (R_{ij} \delta^h_{k} - R_{ik} \delta^h_{j}),
\end{eqnarray}
so that the expression for the divergence of the projective curvature tensor is
\begin{eqnarray}\label{2.11n}
	{P}^h_{ijk;h} = R^h_{ijk;h} -  \frac{1}{3} (R_{ij;k} - R_{ik;j} ),
\end{eqnarray}
which on using equation \eqref{2.7n} leads to
\begin{eqnarray}\label{2.12n}
	{P}^h_{ijk;h} =  \frac{2}{3} (R_{ij;k} - R_{ik;j}).
\end{eqnarray}
By hypothesis the space is of constant curvature. Then
\begin{equation}
	{R}^h_{ijk} = \lambda [\delta ^h_{k} g_{ij} - \delta ^h_{j} g_{ik}],
\end{equation}
which implies by contraction $R_{ij} = 3 \lambda g_{ij},$  ${R}^h_{ijh} = R_{ij}.$

\noindent That is, the spacetime  is an Einstein spacetime and the scalar curvature is a constant.
Therefore, $R_{ij;k} = 0 $ and $R_{,l} = 0.$

\noindent Hence, from the expression  \eqref{2.5n3} and  \eqref{2.12n} of the divergence of the M-projective curvature tensor and projective curvature tensor, respectively, both vanish. So they are identical. 
 Thus, the theorem is established.
\end{proof}

\begin{remark}
	It is clear from equation \eqref{2.12n} that the Ricci tensor must be of the Codazzi type in order for the divergence of the projective curvature tensor to zero. Additionally, for both Einstein spaces and Ricci flat spaces, the divergence of the projective curvature tensor vanishes.
\end{remark}

\begin{theorem}
For  a spacetime  $(V_4, g)$ with constant curvature, the divergences of M-projective curvature tensor and Weyl conformal tensor are identical.	
\end{theorem}

\begin{proof}
Weyl \cite{Weyl} defined the conformal curvature tensor $C^h_{ijk}$, often known as the Weyl conformal curvature tensor, by the equation
	\begin{equation}\label{2.15n}
	C^h_{ijk}= R^h_{ijk}+\frac{1}{2}(\delta^h_j R_{ik}-\delta^h_k R_{ij}+ R^h_j g_{ik}- R^h_k g_{ij})+\frac{R}{6}
	(g_{ij}\delta^h_k- g_{ik}\delta^h_j).
\end{equation}
 Using equation \eqref{2.15n} covariant derivative to get the following as the expression for the divergence of the conformal curvature tensor
\begin{equation}\label{2.16n}
	C^h_{ijk;h} =    R^h_{ijk;h}  + \frac{1}{2} (  R_{ik;j} - R_{ij;k} )+  \frac{2}{3} ( g_{ik} R_{,j}  - g_{ij} R_{,k} ).
\end{equation} 
  The divergence of the M-projective curvature tensor and the conformal curvature tensor are connected by the equations \eqref{2.5n3} and \eqref{2.16n}, 
 \begin{equation}\label{2.17n}
 	C^h_{ijk;h} =    \bar{W}^h_{ijk;h}  + \frac{2}{3} (  R_{ik;j} - R_{ij;k} )+  \frac{1}{2} ( g_{ik} R_{,j}  - g_{ij} R_{,k} ).
 \end{equation}
This completes the proof.
\end{proof}
 
 \begin{theorem}
 For  spacetime  $(V_4, g)$ with constant curvature, the divergences of M-projective curvature tensor and  conharmonic curvature tensor are identical.	
 \end{theorem}

 \begin{proof}
The conharmonic curvature tensor \cite{Ishii Y} $Z^h_{ ijk}$ is defined for a $V_ 4$ as
\begin{equation}\label{2.18n}
	Z^h_{ijk}=R^h_{ijk}+\frac{1}{2}(\delta^h_{j} R_{ik}-\delta^h_{k} R_{ij}+g_{ik} R^h_j-g_{ij} R^h_k).
\end{equation}
The divergence of $Z^h_{ijk}$ is determined by taking the covariant derivative of equation \eqref{2.18n}
\begin{equation}\label{2.19n}
	Z^h_{ijk;h} =    R^h_{ijk;h}  + \frac{1}{2} (  R_{ik;j} - R_{ij;k}  +  g_{ik} R_{,j}  - g_{ij} R_{,k} ).
\end{equation} 

\noindent The divergence of the M-projective curvature tensor and the conharmonic curvature tensor are connected by the equations \eqref{2.5n3} and \eqref{2.19n}, as follows
\begin{equation}\label{2.20n}
	Z^h_{ijk;h} =    \bar{W}^h_{ijk;h}  + \frac{1}{3} (  R_{ik;j} - R_{ij;k}  +  g_{ik} R_{,j}  - g_{ij} R_{,k} ).
\end{equation} 
This completes the proof.
\end{proof}
 \begin{theorem}
 The divergences of M-projective curvature tensor and  concircular curvature tensor are identical in a $(V_4, g)$ with constant curvature.	
 \end{theorem}
 
 \begin{proof}
The Concircular curvature tensor $\bar{C}^h_{ijk}$, for a $V_4$  is defined \cite{AhsanSA} as 
\begin{equation}\label{2.22n}
	\bar{C}^h_{ijk}=R^h_{ijk} - \frac{R}{12}(\delta^h_{k} g_{ij}-\delta^h_{j} g_{ik}).
\end{equation}
The divergence of $\bar{C}^h_{ijk}$ is determined by taking the covariant derivative of equation \eqref {2.22n}
\begin{equation}\label{2.23n}
	\bar{C}^h_{ijk;h} =    {R}^h_{ijk;h}  - \frac{1}{12} (   g_{ij} R_{,k}  - g_{ik} R_{,j} ).
\end{equation} 
Using equations \eqref{2.5n3} and \eqref{2.23n}, we get
\begin{equation}\label{2.24n}
	\bar{C}^h_{ijk;h} =    \bar{W}^h_{ijk;h}  + \frac{1}{6} (  R_{ij;k} - R_{ik;j} )+  \frac{1}{12} ( g_{ij} R_{,k}  - g_{ik} R_{,j} ).
\end{equation}
This completes the proof.
\end{proof}
 \begin{theorem}
 For  spacetime  $(V_4, g)$ with constant curvature, the divergences of M-projective curvature tensor and  W-curvature tensor are identical.
 \end{theorem}
	
\begin{proof}
A W-curvature tensor, also known as a $W_2$-curvature tensor, was first introduced by Pokhariyal and Mishra \cite{GPPokh}  in 1970 and its definition is
\begin{equation}\label{2.25n}
	W^h_{ijk} = R^h_{ijk} + \frac{1}{3} [R_{ik} \delta^h_{j} -  g_{ij} R^h_{k}].
\end{equation}
The divergence of $W^ h_{ ijk}$ is determined by taking the covariant derivative of equation \eqref {2.25n},
\begin{eqnarray}\label{2.26n}
	W^h_{ijk;h} = R^h_{ijk;h} + \frac{1}{3} [R_{ik;j}  -  g_{ij} R_{,k}].
\end{eqnarray}
\noindent The divergence of the M-projective curvature tensor and the W-curvature tensor are connected by the equations \eqref{2.5n3} and \eqref{2.26n}, as follows
\begin{equation}\label{2.27n}
	W^h_{ijk;h} =    \bar{W}^h_{ijk;h}  + \frac{1}{6} (  R_{ij;k} + R_{ik;j} -  g_{ij} R_{,k}  - g_{ik} R_{,j} ).
\end{equation} 
This completes the proof.
\end{proof}

\begin{remark}
Since a space with constant curvature is an Einstein space, it follows from the foregoing discussions (cf. Theorems 2.2 to 2.6) that for Einstein spaces, the divergence of the M-projective curvature tensor is the same as the divergence of the projective, conformal, conharmonic, concircular, and W-curvature tensors, despite the fact that all six curvature tensors have different.
\end{remark}

\section{\small M-projective Curvature Collineation}
In this section, we will define a symmetry property of spacetime of general relativity in terms of the M-projective curvature tensor and determine the prerequisites followed for  $(V_4, g)$ to admit for the presence of a certain symmetry. The following definitions are important for the development of the results in this section.
\begin{definition}
	A spacetime  $(V_4, g)$   is said to admit M-projective curvature collineation along a  vector field $\xi^{i}$  such that
	\end{definition}
\begin{equation}\label{3.1}
	 \pounds_\xi \bar{W}^h_{ijk} = 0,
\end{equation}
 where $\bar{W}^h_{ijk} $ is M-projective curvature tensor defined through the equation \eqref{2.5n} (For detailed study of Lie derivatives and collineations see \cite{yanoK}).
 Taking Lie derivative of Equation \eqref{2.5n} with respect to vector field $\xi$ 
 \begin{eqnarray}\label{3.2n}
 \pounds_\xi	\bar{W}^h_{ijk} = \pounds_\xi  R^h_{ijk} - \frac{1}{6} \pounds_\xi [R_{ij} \delta^h_{k} - R_{ik} \delta^h_{j} + g_{ij} R^h_{k}  - g_{ik} R^h_{j}].
 \end{eqnarray}
 Katzin et al. \cite{Katzin} have given the relationship chart of different symmetry properties of a spacetime of general relativity and from that chart we deduce have following lemma:
 \begin{lemma}\label{Lemma1}
\cite{Ahsan2} Every motion in $V_n$ implies Weyl projective collineation (WPC), Weyl conformal collineation (W conf C), and curvature collineation (CC).

 \end{lemma}

\begin{theorem}
If a spacetime $(V_4, g)$ admits motion, then M-projective curvature collineation admit along a vector field $\xi$. 
\end{theorem}
\begin{proof}
The Riemann tensor is defined as follows in terms of Weyl projective curvature tensor $P^h_ {ijk}$,
\begin{eqnarray}\label{3.3}
	{R}^h_{ijk} = P^h_{ijk} +  \frac{1}{3} (R_{ij} \delta^h_{k} - R_{ik} \delta^h_{j}).
\end{eqnarray}
Using equations \eqref{3.3} and \eqref{2.5n}, the M-projective curvature tensor expression is to be written as
\begin{eqnarray}\label{3.4}
	\bar{W}^h_{ijk} = P^h_{ijk} + \frac{1}{6} [R_{ij} \delta^h_{k} - R_{ik} \delta^h_{j} - g_{ij} R^h_{k}  + g_{ik} R^h_{j}].
\end{eqnarray}
The Weyl conformal tensor can be written similarly,
\begin{equation}\label{3.5}
	\bar{W}^h_{ijk} =    C^h_{ijk}  -  \frac{R}{6} (  \delta^h_{k} g_{ij} - \delta^h_{j} g_{ik} )+  \frac{1}{3} (R_{ij} \delta^h_{k} - R_{ik} \delta^h_{j} + g_{ij} R^h_{k}  -  g_{ik} R^h_{j}),
\end{equation}
where $R = g_ {ij} R^ij$ is scalar curvature and $C^h_{ijk}$ is the Weyl conformal curvature tensor. 
\vspace{6pt}
\par

It is evident from equations \eqref{3.4} and \eqref{3.5}, \cite{Katzin} every motion implies to CC and every CC implies to RC, therefore we can then equate the  Lie derivative of the M-projective curvature tensor with that of the Weyl projective and conformal curvature tensor. Equation \eqref{3.4} or \eqref{3.5} provides the following by using Lemma \ref{Lemma1}. This completes the proof. 
\end{proof}{\small {\tiny }}
Now consider the following lemma:
\begin{lemma}\label{Lemma2}
	\cite{Michalaski } If $\xi$ is a Killing vector in a non-null electromagnetic field, then the Lie derivative of the electromagnetic field tensor $F_{ij}$ with respect to $\xi$  vanishes.
	\end{lemma}
\begin{theorem}
A non-null electromagnetic field admits M-projective curvature collineation if it admits motion.
\end{theorem}
\begin{proof}
The energy momentum tensor $T_{ij}$ for a non-null electromagnetic field is written as

\begin{equation}\label{3.6}
	T_{ij} =   F_{im} F^{m}_{j}  - \frac{1}{4}   g_{ij} F^{ab}F_{ab}. 
\end{equation} 
Considering the purely electromagnetic distribution of the Einstein field equations
\begin{equation}\label{3.7}
	R_{ij} =   \kappa T_{ij},  
\end{equation} 
where $\kappa$ is gravitational constant. Equation \eqref{3.5}, which incorporates equations \eqref{3.7} and \eqref{3.6}, yields

\begin{eqnarray}\label{3.8}
	\bar{W}^h_{ijk} =    C^h_{ijk}  -  \frac{R }{6} (  \delta^h_{k} g_{ij} - \delta^h_{j} g_{ik} )  +  \frac{\kappa}{3} (F_{im} F^{m}_{j} \delta^h_{k} - F_{im} F^{m}_{k} \delta^h_{j}  + g_{ij} F^h_{m} F^{m}_{k}  -  g_{ik} F^h_{m} F^{m}_{j} ) \notag\\
	-  \frac{\kappa}{12} ( g_{ij} \delta^h_{k}  - g_{ik} \delta^h_{j} ) F_{ab} F^{ab}   -   \frac{\kappa}{12} ( g_{ij} g^h_{k}   -  g_{ik} g^h_{j}) F_{ab} F^{ab}.
\end{eqnarray}
Applying Lemmas \ref{Lemma1} and \ref{Lemma2} in equation \eqref{3.8}, then establish the Theorem. 
\end{proof}

Equations \eqref{3.4}, \eqref{3.5}, and \eqref{3.6} as well as Lemmas \ref{Lemma1} and \ref{Lemma2} can also yield a similar conclusion.
\vspace{6pt}
\par
For a null electromagnetic, the energy-momentum tensor is given by
\begin{equation}\label{3.9}
	T_{ij} =   F_{in} F^{n}_{j},  
\end{equation} 
where  $ F_{ij} = s_i t_j - t_i s_j$ and $s^i s_i = s^i t_ i = 0$, $t^i t_i = 1$. The propagation and polarization vectors are $s$ and $t$, respectively.
As a result of applying equations \eqref{3.7} and \eqref{3.9} to equation \eqref{3.5}, we obtain
\begin{eqnarray}\label{3.10}
	\bar{W}^h_{ijk} =    C^h_{ijk}  -  \frac{R }{6} (  \delta^h_{k} g_{ij} - \delta^h_{j} g_{ik} )  +  \frac{\kappa}{3} (F_{in} F^{n}_{j} \delta^h_{k} - F_{in} F^{n}_{k} \delta^h_{j}\notag\\
	  + g_{ij} F^h_{n} F^{n}_{k}  -  g_{ik} F^h_{n} F^{n}_{j}). 
\end{eqnarray}
The fact that,

\begin{lemma}\label{Lemma3}
\cite{Zahsan2}	If the propagation (polarization) vector is Killing and expansion-free, then  a null electromagnetic field admits Maxwell collineation along the propagation (polarization) vector.
\end{lemma}

The Lie derivative of M-projective curvature tensor with regard to propagation (polarization) vector vanishes from equation \eqref{3.10} and Lemmas \ref{Lemma1} and \ref{Lemma3}. As a result, we may state that

\begin{theorem}
If a propagation (polarization) vector is Killing and expansion-free, then a null electromagnetic field admits M-projective curvature collineation along a vector $\xi$.	
\end{theorem}

\begin{remark}
As M-projective curvature tensor may be written in terms of other curvature tensors like concircular curvature tensor and conharmonic curvature tensor, several equivalent conclusions can be achieved for \cite{AhsanSA} M-projective curvature collineation.
\end{remark}

\section{ Conclusion}
This study has made an effort to look at the connection between the divergence of the M-projective curvature tensor and other curvature tensors, including projective, conformal, concircular, conharmonic, and W-curvature curvature tensor. Additionally, we have explored the circumstances in which the spacetime of general relativity may admit M-projective collineation, as well as introducing the idea of M-projective collineation. There were instances of a non-null and a null electromagnetic field. The Conformal Killing vector field has many applications in  astrophysics, like star modeling, wormhole modeling etc \cite{CKV,CKV1,CKV2,CKV3}. Similarly, if we  can write curvature collineation, curvature inheritance, and Ricci inhertance etc in terms of local coordinates, that is,  in spherically symmetric   static spacetime, then, several astrophysical models can be discussed. May be that will be our  future  project.

\par
\vspace{6pt}
{\bf Acknowledgments}\quad Farook Rahama and Naeem Ahmad Pundeer would like to thank the authorities of the Inter-University Centre for Astronomy and Astrophysics, Pune, India for providing research facilities. Also, the fourth author is supported by U.G.C Dr. D.S. Kothari Post Doctoral fellowship No. F.4-2/2006 (BSR)/MA/20-21/0069.

\end{document}